\documentclass[11pt]{article}
\usepackage{collect,fullpage}

\usepackage{amsmath,amssymb,amsfonts,amsthm}
\usepackage[pdfauthor={Krishnamoorthy Dinesh, Jayalal Sarma},
pdftitle={Alternation, Sparsity and Sensitivity :  Bounds and Exponential Gaps},
pdfkeywords={Boolean function parameters, Sensitivity Conjecture, XOR Log-Rank Conjecture},
citecolor=blue,linkcolor=blue,colorlinks=true]{hyperref}
\usepackage{caption}
\usepackage{verbatim,graphicx}
\usepackage{complexity}

\usepackage[capitalise] {cleveref}	
\usepackage{xspace}

\usepackage{todonotes}

\usepackage{tikz}
\usepackage{subcaption}

\theoremstyle{plain}
\newtheorem{theorem}{Theorem}[section]
\newtheorem{lemma}[theorem]{Lemma} 
\newtheorem{proposition}[theorem]{Proposition}

\newtheorem{corollary}[theorem]{Corollary}

\newtheorem{remark}[theorem]{Remark}

\theoremstyle{definition}
\newtheorem{definition}[theorem]{Definition}

\newcounter{todocounter}

\newcommand{\introthm}[2]{\vspace{1mm}\noindent \textbf{\cref{#1}} \textit{#2} \vspace{1mm}}

\newcommand{\sse}{\subseteq}
\newcommand{\zo}{\{0,1\}}
\newcommand{\zon}{\zo^n}
\newcommand{\defn}{\stackrel{\text{\tiny def}}{=}}
\newcommand{\degtwo}{\mathsf{deg}_2}
\newcommand{\rank}{\mathsf{rank}}

\newcommand{\fhat}{\widehat{f}}
\newcommand{\ghat}{\widehat{g}}
\newcommand{\pmo}{\set{-1,1}}
\newcommand{\pmon}{\pmo^n}
\newcommand{\I}{\mathbf{I}}
\newcommand{\Expt}{\mathbf{E}}
\newcommand{\sps}{\mathsf{sparsity}}

\newcommand{\card}[1]{\left | \set{#1}\right |}
\newcommand{\set}[1]{\left\{ #1 \right\}}

\newcommand{\etal}{{et al}.}

\newcommand{\DT}{\textsf{DT}}

\newcommand{\sens}{\mathsf{s}}
\newcommand{\bsens}{\mathsf{bs}}
\newcommand{\cert}{\mathsf{C}}
\newcommand{\alt}{\mathsf{alt}}
\renewcommand{\deg}{\mathsf{deg}}
\newcommand{\dc}{\mathsf{dc}}
\newcommand{\negs}{\mathsf{negs}}

\newcommand{\calA}{{\cal A}}
\newcommand{\calB}{{\cal B}}
\newcommand{\calC}{{\cal C}}

\newcommand{\calF}{{\cal F}}
\newcommand{\calG}{{\cal G}}
\newcommand{\calM}{{\cal M}}

\newcommand{\ADDR}{\mathsf{ADDR}}

\newcommand{\CC}{\mathsf{CC}}
\newcommand{\XOR}{{\sc Xor}}

\newcommand{\F}{{\mathbb{F}}}

\newcommand{\N}{{\mathbb{N}}}
\newcommand{\Z}{{\mathbb{Z}}}

\renewcommand{\R}{{\mathbb{R}}}
\newcommand{\degm}{\mathsf{deg}_m}

\renewcommand{\deg}{\mathsf{deg}}

\def\movetoappendix{1}

\definecollection{appendix}
\makeatletter
\newenvironment{aproof}[2]
  { \@nameuse{collect}{appendix}
  { \subsection{#1} \label{#2} \begin{proof} } {\end{proof}}
  }{\@nameuse{endcollect}}
\makeatother

\makeatletter
\newenvironment{appsection}[2]
  { \@nameuse{collect}{appendix}
  { \subsection{#1} \label{#2} }
  {}
  }{\@nameuse{endcollect}}
\makeatother

\ifthenelse{\equal{\movetoappendix}{0}}{
        
        \renewenvironment{appsection}[2]{} {}
}{}

\title{Alternation, Sparsity and Sensitivity : \\ Bounds and Exponential Gaps}
\author{Krishnamoorthy Dinesh \and  Jayalal Sarma}
\date{
Department of Computer Science \& Engineering \\
Indian Institute of Technology Madras, Chennai, India \\
{\small \texttt{\{kdinesh,jayalal\}@cse.iitm.ac.in}} \\
\vspace{1mm}
\vspace{-1mm}
}
\begin{document}
\maketitle

\begin{abstract} 
	The well-known \textsc{Sensitivity Conjecture} regarding combinatorial
	complexity measures on Boolean functions states that for any Boolean
	function $f:\zon \to \zo$, block sensitivity of $f$ is polynomially
	related to sensitivity of $f$ (denoted by $\sens(f)$). From the
	complexity theory side, the \textsc{\XOR~Log-Rank Conjecture} states that for
	any Boolean function, $f:\zon \to \zo$ the communication complexity of
	a related function $f^{\oplus}:\zon \times \zon \to \zo$ (defined as
	$f^{\oplus}(x,y) = f(x \oplus y)$) is bounded by polynomial in
	logarithm of the sparsity of $f$ (the number of non-zero Fourier
	coefficients for $f$, denoted by $\sps(f)$). Both the conjectures play
	a central role in the domains in which they are studied.

A recent result of Lin and Zhang~(2017) implies that to confirm the above two
	conjectures it suffices to upper bound alternation of $f$ (denoted by $\alt(f)$) for all Boolean functions $f$
	by polynomial in $\sens(f)$ and logarithm of $\sps(f)$, respectively. 
	In this context, we show the following results:

\begin{itemize}
\item We show that there exists a family of Boolean functions for which
	$\alt(f)$ is at least \textit{exponential} in  $\sens(f)$ and
		$\alt(f)$ is at least \textit{exponential} in
		$\log\sps(f)$. En route to the proof, we also show an exponential gap
		between $\alt(f)$ and the decision tree complexity of $f$,
		which might be of independent interest.
\item As our main result, we show that, despite the above exponential gap between $\alt(f)$ and $\log\sps(f)$, the \textsc{\XOR~Log-Rank Conjecture} is true for functions with the alternation upper
	bounded by $\poly(\log n)$. It is easy to observe that the \textsc{Sensitivity Conjecture}
	is also true for this class of functions.  
\item  The starting point for the
	above result is the observation (derived from Lin and Zhang (2017)) that for
	any Boolean function $f$, $\deg(f) \le \alt(f)\degtwo(f)\degm(f)$
	where $\deg(f)$, $\degtwo(f)$ and $\degm(f)$ are the degrees of $f$ over $\R$, $\F_2$ and $\Z_m$ respectively. We
	give three further applications of this bound: \textsf{(1)}~We show that for Boolean functions $f$ of constant alternation  have $\degtwo(f) = \Omega(\log n)$.   \textsf{(2)}~Moreover, these functions also have high
	sparsity  ($\exp(\Omega(\sqrt{\deg(f)})$), thus partially answering
	a question of Kulkarni and Santha (2013). \textsf{(3)}~We
	observe that our relation upper bounding real degree also improves the upper
	bound for influence to $\degtwo(f)^2 \cdot \alt(f)$ improving Guo and 
	Komargodski (2017).
\end{itemize}

\end{abstract}

\newpage
\section{Introduction}
A central theme of research in Boolean function complexity is relating the 
complexity measures of Boolean functions (see~\cite{BW02} for a survey). For a Boolean function $f:\zon \to \zo$, \emph{sensitivity} of $f$ on $x \in \zon$, is the maximum number of indices $i \in [n]$, such that $f(x \oplus e_i) \ne f(x)$ where $e_i \in \zon$ with exactly the $i^{th}$ bit as $1$. The sensitivity of 
$f$ (denoted by $\sens(f)$) is the maximum sensitivity of $f$ over all inputs. A 
related parameter is the \emph{block sensitivity} of $f$ (denoted by $\bsens(f)$), where we allow disjoint 
blocks of indices to be flipped instead of a single bit. Another parameter is the 
\emph{degree} (denoted by $\deg(f)$) 
of a multilinear polynomial over reals that agrees with $f$ on Boolean 
inputs. If the polynomial is over $\F_2$, then the degree of the polynomial is 
denoted by $\degtwo(f)$.

 Nisan and Szegedy conjectured that for an arbitrary function $f:\zon \to \zo$, $\bsens(f)  
 \le \poly(\sens(f))$ and this is popularly known as the \emph{\textsc{Sensitivity Conjecture}}
 \cite{NS92}. 
Though the measures $\bsens(f)$ and $\sens(f)$ were introduced to understand
the {\sc Crew-Pram} model of computation~\cite{CDR86,N91}, subsequent works \cite{BW02,BBCMW01,Tal13}  
showed connection to other Boolean function parameters, in particular, $\sqrt{\bsens(f)} \le$
 $\deg(f) \le \bsens(f)^3$.  Hence the \textsc{Sensitivity Conjecture} 
can equivalently be stated as : for any Boolean function $f$, $\deg(f) \le \poly(\sens(f))$. 
This question has been 
extensively studied in \cite{Sim83,KK04,APV16} (see~\cite{HKP11} for a survey) and for various restricted classes of Boolean functions in~\cite{N91,Tur84,Cha05,ST16,BLTV16}. 
There are also recent approaches to settle the conjecture via a formulation in terms of a communication game~\cite{GKS15} and via a formulation in terms of distributions on the Fourier spectrum of Boolean functions~\cite{GSW16}.

Moving on to a computational setting, for a function $F:\zon \times \zon \to \zo$, consider two parties Alice and Bob having  $x \in \zon$ and  $y \in \zon$ respectively. The goal is to come up with a deterministic protocol to compute $F(x,y)$ while minimizing the number of bits exchanged. We call the number of bits exchanged over all inputs in the worst case  as the cost of the protocol. A natural question is to understand cost of the minimum cost protocol computing $F$ (denoted by $\CC(F)$). Lovasz and Saks~\cite{LS93} related this to the rank over reals of a matrix $M_F$ defined as $[F(x,y)]_{(x,y) \in \zon \times \zon}$ by showing that $\CC(F)$ is at least $\log \rank(M_F)$ and they conjectured that this is tight up to polynomial factors. 
More precisely, they conjectured that for all $F$, $\CC(F)$ is at most polynomial in $\log \rank(M_F)$. Popularly known as the  \emph{\textsc{Log-Rank Conjecture}}, this has been intensely studied in the past 25 years with the best known upper bound being $O(\sqrt{\rank(M_F)} \log \rank(M_F))$ (see~\cite{L16} and the references therein). 
With the conjecture being far from settled in the general case, it is natural to looks at special class of functions $F$.

One such class that can be considered is the XOR functions of the form $F(x,y) = f(x\oplus y)$ where $f$ is an $n$ bit Boolean functions. For such an $F$, the $\rank(M_F)$ is exactly the sparsity of $f$ in its Fourier representation (see~\cref{sec:prelims} for a definition).
For $x,y \in \zon$, with $f^\oplus(x,y) = f(x\oplus y)$,  
the \emph{\textsc{\XOR~Log-Rank Conjecture}} (proposed in \cite{ZS10}) says that the deterministic 
communication complexity of $f^\oplus$ (denoted by $\CC_{\oplus}(f)$) is  upper bounded by a polynomial in logarithm of sparsity of $f$ (denoted by $\sps(f)$). 
The conjecture is known to hold for restricted classes of Boolean functions like monotone
functions \cite{MO09}, symmetric functions~\cite{ZS09}, functions computable by constant depth polynomial size  circuits~\cite{KS13},  functions of small spectral norm~\cite{TWXZ13} and read restricted polynomials over $\F_2$~\cite{CW18}.

Recently, Lin and Zhang~\cite{LZ17} studied both the above stated conjectures in connection 
to alternation, a measure of non-monotonicity of $f$ (denoted by $\alt(f)$) which, stated informally, is the maximum number of times the value of $f$ changes along any chain on the Boolean cube from $0^n$ to $1^n$ (see~\cref{sec:prelims} for  a definition). They showed that for any Boolean function $f$, $\bsens(f) 
= O(\sens(f)\alt(f)^2)$ and $\CC_{\oplus}(f)= O(\log \sps(f) \alt(f)^2)$. These results shows 
that to settle the \textsc{Sensitivity Conjecture}, it suffices to show that for any Boolean 
function $f$, $\alt(f) \le \poly(\sens(f))$ and to settle the \textsc{\XOR~Log-Rank 
Conjecture}, it suffices to show that for any Boolean function $f$, $\alt(f) \le \poly(\log 
\sps(f))$. 
\paragraph{{\bf Our Results :}} As a first step, we ask whether it is indeed true that for all Boolean functions $f$, $\alt(f)=O(\poly(s(f)))$  and   for all Boolean functions $f$, $\alt(f) \le \poly(\log \sps(f))$. 
We answer both of these questions in the negative by exhibiting a family of Boolean functions $\calF = \set{f_k \mid  k \in \N}$ (\cref{def:gap-fun}) for which $\alt(f_k)$ is at least \emph{exponential} in $\sens(f_k)$ and $\alt(f_k)$ is at least \emph{exponential} in $\log \sps(f_k)$.

\begin{theorem}\label{gap:exp}
There exists a family of Boolean functions $\calF = \{f_k:\zo^{n_k} \to \zo$ $\mid k \in \N\}$ such that 
$\alt(f_k) \ge 2^{\sens(f_k)}-1$ and $\alt(f_k) \ge 2^{(\log \sps(f_k))/2}-1$.
\end{theorem}

The main property of $f_k \in \calF$ which we exploit to prove~\cref{gap:exp} is that $
\alt(f_k) = 2^{\DT(f_k)}-1$ (\cref{newgap:alt:dt}) where $\DT(f_k)$ is the depth of the 
optimal decision tree computing $f_k$ (see~\cref{sec:prelims} for definition). We also show 
an asymptotically matching upper bound for alternation of any Boolean function. More 
precisely, for any $f:\zon \to \zo$, 
we show that  $ \alt(f) \le 2^{\DT(f)+1}-1$ (\cref{alt:dt}).

Though the function family $\calF$ rules out settling the \textsc{Sensitivity Conjecture} 
(\textsc{\XOR~Log-Rank Conjecture} resp.) via upper 
bounding alternation by a polynomial in sensitivity (polynomial in the logarithm of sparsity 
resp.) for all Boolean functions, it is partly unsatisfactory since both the conjectures 
holds for functions with shallow decision trees computing them and the Boolean functions in $\calF$, by definition, has such a property (see~\cref{sec:compose} for details).

In fact, any $f:\zon \to \zo$ for which $\alt(f) = 2^{\Omega(\DT(f))}$, must satisfy $ \DT(f) = O(\log n)$. In addition, if $f$ depends on all the input variables, the \textsc{Sensitivity Conjecture} is true for $f$. Notice that, for all $f_k \in \calF$, $\DT(f_k) = \log n_k$ and $f_k$ depends on all the $n_k$ variables.
Hence a natural question is, does there exist another family of functions $f$ where $\alt(f)$ is at least super-polynomial in $\sens(f)$, but $\DT(f)$ is not logarithmic in $n$. To this end, we exhibit
a  family of Boolean functions $\calG$, such that for all $g \in \calG$, $\alt(g)$ is \textit{super-linear} in $\sens(g)$ and $\DT(g)$ is $\omega(\log n)$ where $n$ is the number of variables in $g$. 
\begin{theorem}\label{alt:sens:sep}
There exists a family of Boolean functions $\calG = \{ g_k :\zo^{n_k} \to \zo$ $\mid k \in \N\}$ such 
that $\alt(g_k) \ge \sens(g_k)^{\log_3 5}$ while $\DT(g_k)$ is $\Omega(n_k^{\log_6 3})$. 
\end{theorem}

\noindent
The main tool used in proving~\cref{alt:sens:sep} is a bound on the alternation of composed 
Boolean functions (\cref{alt:composition}). 

As mentioned before, Lin and Zhang~\cite{LZ17} showed that \textsc{\XOR~Log-Rank Conjecture} is true 
for all Boolean functions satisfying $\alt(f) \le \poly(\log \sps(f))$. As our main result,  we further strengthen this when $\sps(f) < n$.
\begin{theorem}[Main] \label{alt:polylog}
For large enough $n$, the \textsc{\XOR~Log-Rank Conjecture} is true for all $f:\zon \to \zo$, such that  $
\alt(f) \le \poly(\log n)$ where $f$ depends on all its inputs.
\end{theorem}

\noindent
Let $\deg_m(f)$ be the degree of the polynomial agreed with $f$ on $\zon$ over the ring $\Z_m$. Our starting point in proving the above result is a relation connecting $\deg$, $\degtwo$, $\degm$ and $\alt$.  For all Boolean functions $f$,
\begin{equation}
\deg(f) \le \alt(f) \cdot \degtwo(f)\cdot \degm(f) \label{eq:deg:deg2}
\end{equation}
We remark that for a special cases, the above bound~(\ref{eq:deg:deg2}) on $\deg(f)$ is known to be true. For instance, if 
$f$ is a monotone, it can be shown\footnote{When $f$ is monotone, it is known that $\DT(f) \le \sens(f)^2$ and $\sens(f) \le \degtwo(f)$ (Corollary 5 and Proposition 4  resp., of~\cite{BW02}). Proposition 4 of~\cite{BW02} though states that $\sens(f) \le \deg(f)$ for any monotone $f$, the argument is also valid for $\degtwo(f)$. Since, $\deg(f) \le \DT(f)$ (see for instance~\cite{BW02}), $\deg(f) \le \degtwo(f)^2$.} that $\deg(f) \le \degtwo(f)^2$. 
However, there are functions of large alternation where $\degtwo(f)$ is constant while $\deg(f)$ is $n$ (for instance, parity on $n$ bits). 
Hence, we cannot upper bound degree by $\F_2$-degree in general but~\cref{eq:deg:deg2} (for $m=2$) says that we can indeed upper bound $\deg(f)$ by $\degtwo(f)$ using $\alt(f)$. This case ($m=2$) is implicit in~\cite{LZ17}. 
We now give three further applications of~\cref{eq:deg:deg2} (see~\cref{sec:deg-deg2:app}).

As our first application, we show that using a related result due to Gopalan \etal~\cite{GLS09} and~\cref{eq:deg:deg2} (with $m=2$), all Boolean functions of bounded alternation ($\alt(f) = O(1)$) must have $\degtwo(f) = \Omega(\log n)$  
(see~\cref{corr:deg2lb} for a generalization).

As a second application, we show that Boolean functions with bounded alternation have high sparsity.
Kulkarni and Santha~\cite{KS13} had studied the relation between $\log \sps(f)$ and $\deg(f)$ in the case of restricted families of monotone  functions and asked if they are linearly related in the case of monotone functions. In this direction, we show the following lower bound for $\log_2 \sps(f)$ in terms of $\deg(f)$.
\begin{theorem}\label{alt:bound}
For Boolean functions $f$ with $\alt(f) = O(1)$, $\log \sps(f) = \Omega(\sqrt{\deg(f)})$.
\end{theorem}
  
As a third application, we observe that~\cref{eq:deg:deg2} implies an improved upper bound for influence (denoted by $\I[f]$, see~\cref{sec:prelims} for a definition) to $\degtwo(f)^2 \cdot \alt(f)$. This improves the result of Guo and Komargodski~\cite{GK17} who showed that $\I[f] = O(\alt(f) \sqrt{n})$, thus giving faster learning algorithms for functions of bounded alternation in the PAC learning model.

\section{Preliminaries}
\label{sec:prelims}
We introduce the notations and definitions used in this paper. All logarithms are to the base $2$ unless otherwise stated.
Let $[n] \defn \set{1,2,\dots,n}$. For $i \in [n]$, define $e_i$ to be the $n$ bit Boolean string with  one in $i^{th}$ location and zero elsewhere. 
A Boolean function $f:\zon \to \zo$ is \emph{monotone} if 
$\forall x,y \in \zon$, $x
\prec y \implies f(x) \le f(y)$, where $x \prec y$ iff $\forall i \in [n], x_i \le y_i$. The 
\emph{alternation} of a Boolean function is a measure of non-monotonicity of the Boolean 
function. More precisely, if we
define a collection of distinct inputs $x_0,x_1,x_2 \dots, x_n \in \zon$ satisfying $0^n = x_0 
\prec x_1 \prec x_2 \prec \dots \prec x_n = 1^n$ as a
\emph{chain} in the Boolean hypercube $\calB_n$ then, alternation of $f$ (denoted by $\alt(f)
$) is defined as  $\max\set{\alt(f,\calC) \mid
\calC \text{ is a chain in } \calB_n}$ where 
$\alt(f, \calC)$) is $|\{i \mid f(x_{i-1}) \ne f(x_{i}), x_i \in$
$\calC, i \in [n]\}|$. Indeed, for a monotone $f$, $\alt(f) = 1$. 

By definition, any chain $\calC$ of a Boolean hypercube $\calB_n$ is maximal and hence is uniquely determined by a permutation 
$\sigma \in S_n$ and vice versa. An	$x \in \zon$ belongs to a chain $\calC$ 
defined by $\sigma \in S_n$ iff $x=0^n$ or $ x = \bigvee_{i=1}^{wt(x)} e_{\sigma(i)} $ where the OR  
is taken coordinate wise and $wt(x)$ is the number of ones in $x$. If a chain is defined using a permutation $\sigma$, 
we use $\sigma$ to denote the chain $\calC$.	

For a Boolean function $f$ on $m$ variables and $g$ on $n$ variables, we denote $f \circ g$ 
as a function on $mn$ variables $\set{x_{11},\dots, x_{mn}}$ defined as $f(g(x_{11},
\dots,x_{1n}),$ $ g(x_{21},\dots,x_{2n}),\dots, g(x_{m1},\dots,$ $x_{mn}))$. We define 
$g^{\circ k}$ as the Boolean function on $n^k$ variables as  $g^{\circ (k-1)}\circ g$ for $k 
> 1$ and $g$ for $k=1$.

Given a Boolean function, there always exists a unique $n$ variable multilinear
polynomial over $\R[x_1,\ldots,x_n]$ such that the evaluation agrees with the 
function on $\zon$. The \emph{degree} of function $f$ is the degree of such a 
polynomial (denoted by $\deg(f)$). If we consider the polynomial to be over $\Z_m[x_1,\ldots,x_n]$ for an integer $m > 1$ instead,  we get the \emph{$\Z_m$-degree} of $f$ (denoted by $\degm(f)$).  If $m=2$, it is the $\F_2$-degree of $f$ (denoted by $\degtwo(f)$).

A deterministic Boolean \textit{decision tree} is a rooted tree where the leaves are labeled $0$ or $1$ and non-leaf nodes labeled by a variable having two outgoing edges (corresponding to the value taken by the variable). A decision tree is said to compute a Boolean function $f$, if for all inputs $x$, the path from root to the leaf determined by $x$ is labeled $f(x)$. Define $\DT(f)$ as the depth of the smallest depth decision tree computing $f$.

For an $x,y \in \zon$, we denote by $x
\oplus y$, the input obtained by taking bitwise parity of $x$ and $y$. For $B \sse [n]$, 
$e_B$ denotes the characteristic vector of $B$.  For $f:\zon \to \zo$ and $x
\in \zon$, define the sensitivity of $f$ on $x$ (denoted by $\sens(f,x)$) as $\card{i \mid f(x\oplus e_i) \ne f(x), i \in [n]
}$. We define the block sensitivity of $f$ on $x$ (denoted by $\bsens(f,x)$) as the size of maximal collection of disjoint non-empty sets
$\{B_i\}$ where each $B_i \sse [n]$ in the collection satisfy $f(x \oplus e_{B_i}) \ne f(x)$.
\emph{Certificate complexity} of $f$ on input $x$ (denoted by $\cert(f,x)$) is the
size of the smallest certificate $S \sse [n]$ such that $\forall y \in \zon$
with $y|_S = x|_S \implies f(y) = f(x)$. 
The \emph{sensitivity} of $f$ (denoted by $\sens(f)$) is defined as $\max_{x \in \zon} \sens(f,x)$. The \emph{influence} of a Boolean function $f$ (denoted by $\I[f]$) is defined as $\mathbf{E}_{x \in \zon} [s(f,x)]$. The \emph{block sensitivity} of $f$ (denoted by $\bsens(f)$) is $\max_{x \in \zon} \bsens(f,x)$. Note that $\I[f] \le \sens(f) \le \bsens(f)$. It is also known that $\I[f] \le \deg(f) \le \DT(f)$~(see for instance \cite{BW02}).

For $x \in \zon$ and $S \sse [n]$, define $\chi_S(x) = (-1)^{\sum_{i \in S} x_i}$.
Any $f:\zon \to \pmo$ can be uniquely expressed as $\sum_{S \sse [n]} \fhat(S) 
\chi_S(x)$  where $\fhat(S) \in \R$, indexed by $S \sse [n]$, denotes the \emph{Fourier coefficients} of $f$ which is $\frac{1}{2^n}\sum_x f(x)\chi_S(x)$ (see Theorem 1.1,  \cite{OD14} for more details). The \emph{sparsity} of a Boolean 
function $f$ (denoted by $\sps(f)$) is the number of non-zero Fourier coefficients of 
$f$. 
Observe that we have defined sparsity of a Boolean function only when the range of the function is $\pmo$. For Boolean functions $f$ whose range is $\set{0,1}$, we define sparsity of $f$ to be the sparsity of the function $1-2f$ (whose range is $\pmo$).  

\section{Alternation vs Sensitivity and Alternation vs Logarithm of Sparsity}
In this section, we show that there exists a family of function $\calF = \set{f_k \mid k \in \N}$ with $\alt(f_k)$ is at least exponential in $\sens(f_k)$, $\DT(f_k)$ and $\log \sps(f_k)$ respectively (\cref{sec:gap:lb}). Complementing this, we show that for any Boolean function $f$, $\alt(f)$ can be at most exponential in $\DT(f)$ (\cref{sec:gap:ub}). We prove a bound on the alternation of composed Boolean functions and use it to obtain a family of functions with super-linear gap between alternation and sensitivity with large decision tree depth unlike functions in $\calF$ (\cref{sec:compose}).
\subsection{Exponential Gaps : Alternation vs Decision Tree Depth} \label{sec:gap:lb}
We prove~\cref{gap:exp} in this section. We first show that there exists a family of function $\calF = \set{f_k \mid k \in \N}$ with $
\alt(f_k)$ equals $2^{\DT(f_k)}-1$ (\cref{newgap:alt:dt}). Since for any 
Boolean function $f$, $\sens(f) \le \DT(f)$ (see for instance~\cite{BW02}), we have, $\alt(f_k) = 2^{\DT(f_k)} -1\ge 2^{\sens(f_k)}-1$ and
since for any Boolean function $f$, $\log \sps(f) \le 2\deg(f) \le 2\DT(f)$ (for a proof, see Corollary 2.8, ~\cite{T14} and Proposition 3.16, \cite{OD14} respectively),  we get that for 
$f_k$, $\alt(f_k) = 2^{\DT(f_k)}-1 \ge 2^{0.5\log \sps(f_k)}-1$ thereby proving~\cref{gap:exp}.

Hence, one cannot hope to show that for all Boolean functions $f$, alternation is upper bounded 
polynomially by sensitivity of $f$ or polynomially by logarithm of sparsity of $f$. 
We now define our family $\calF$ of Boolean functions.

\begin{definition}\label{def:gap-fun}
Let $\calF = \set{f_k\mid k \in \N}$ be a family of Boolean functions where for every $k 
\in \N$,  $f_k:\zo^{2^k-1} \to \zo$ is defined by the decision 
tree which is a full binary tree of depth $k$ with each of the $2^k-1$ internal node 
querying a distinct variable and each of the nodes at level $k$ have left leaf child 
labeled $0$ and right leaf child labeled $1$.
\end{definition}
By definition, $\DT(f_k) \le k$. In fact, we claim that $\DT(f_k) = k$. To derive this, it suffices to argue that $\deg(f_k) \ge k$ since $\DT(f) \ge \deg(f)$. The degree lower bound can be seen as follows. Recall from the preliminaries~(\cref{sec:prelims}) that the real representation of any Boolean function is unique. Hence, it suffices to show that there is a real representation of $f$ with a monomial of size $k$ with a non-zero coefficient.  A standard way to obtain a real representation is to interpolate the values of $f$ on the $2^n$ inputs (Section 1.2 in \cite{OD14}). Let $x_i$ be the left most variable at depth $k$ in the decision tree computing $f_k$. In the interpolation, the monomial containing $x_i$ of length $k$ cannot get canceled by another monomial since $x_i$ does not appear elsewhere in the decision tree. Hence $\deg(f_k) \ge k$ and this completes the argument implying that $\deg(f_k) = \DT(f_k) = k$.\\

A Boolean function $f_3 \in \calF$ is described using a decision tree in~\cref{fig:ex}. We remark that the same family of function has also been used in Gopalan \etal~\cite{GSW16} to show an exponential lower bound on tree sensitivity (which they introduce as a generalization of sensitivity) in terms of decision tree depth. Note that, in general, lower bound on tree sensitivity need not imply a lower bound on alternation. For instance, the tree sensitivity of Boolean majority function on $n$ bits is $\Omega(n)$ while alternation is just $1$. 

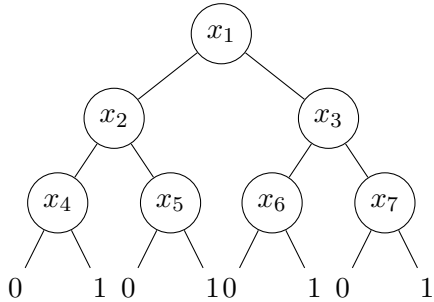
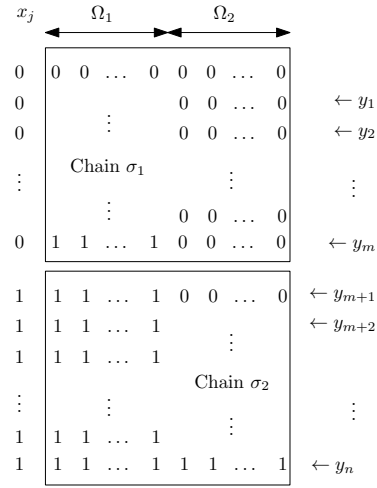
\begin{figure}[htp!]
    \begin{subfigure}[t]{.49\textwidth}
    \centering
	\begin{tikzpicture}[sibling distance=3.8cm, level 2/.style={sibling distance
	=2cm}, level 3/.style={sibling distance=1.5cm},scale=0.75]
	\node[draw, circle] {$x_1$}
	child{ node[circle, draw] {$x_2$} 
		child { node[circle, draw] {$x_4$}
			child { node {$0$} }
			child { node {$1$} }
		}
		child { node[circle, draw] {$x_5$}
			child { node {$0$} }
			child { node {$1$} }
		}
	}
	child{ node[circle, draw] {$x_3$} 
		child { node[circle, draw] {$x_6$}
			child { node {$0$} }
			child { node {$1$} }
		}
		child { node[circle, draw] {$x_7$}
			child { node {$0$} }
			child { node {$1$} }
		}
	};
	\end{tikzpicture}
	\caption{Boolean function $f_3 \in \calF$}
	\label{fig:ex}
	\end{subfigure} \hfill
    \begin{subfigure}[t]{.49\textwidth}
    	\centering
    	\includegraphics[scale=0.72]{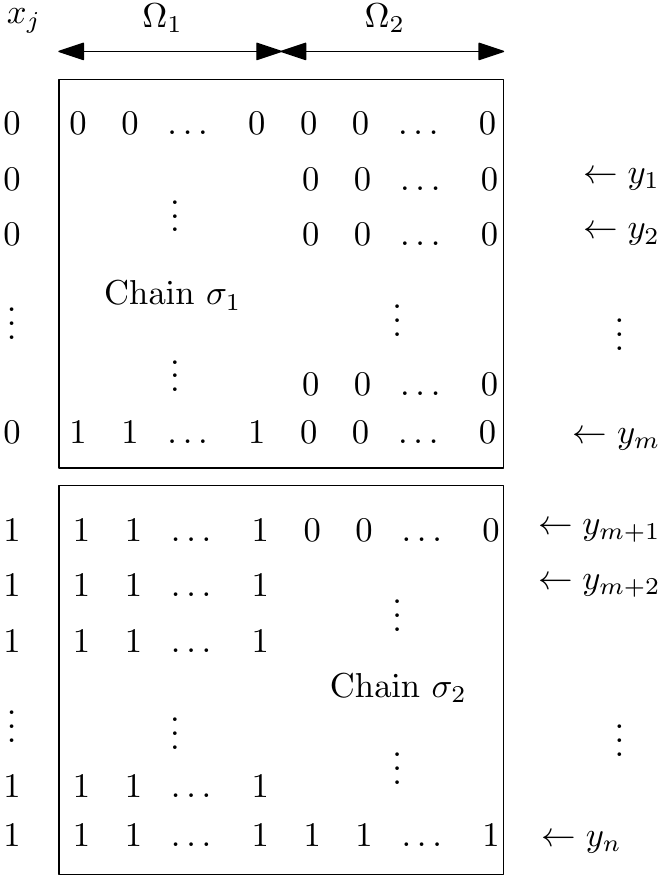}
		\caption{The chain $\sigma$ constructed in the proof of Theorem~\ref{newgap:alt:dt}. 
		Note that $\Omega_1$ and $\Omega_2$ need not be contiguous.}
		\label{fig:proof}
	\end{subfigure}  
	\caption{Boolean function family exhibiting large alternation in Proof of~\cref{newgap:alt:dt}}
\end{figure} 
\begin{theorem}\label{newgap:alt:dt}
For every $k \ge 1$, $f_k \in \calF$, $\alt(f_k) = 2^{\DT(f_k)}-1$.
\end{theorem}

\begin{proof}
By definition, $f_k$ is computed by a decision tree of depth $k$. 
We show that for $k \ge 1$, $\alt(f_k) \ge 2^k-1$. Since $f_k$ is defined on $2^k-1$ variables, we get that $\alt(f_k) = 2^k-1$ thereby completing the proof. 

Since we need to work with functions whose variable set is not necessarily numbered from $1$ to $n$, we associate bijections (instead of permutations) with chains. For any set $\Omega$, let $B_{\Omega}$ be defined as $\set{ \sigma : [|\Omega|] \to 
\Omega \mid \sigma \text{ is a bijection} }$. For a Boolean function $f$ defined on the variables $\set{x_{i_1},x_{i_2},\ldots,x_{i_n} }$, let $var(f)$ be defined as $\set{i_1,i_2, \ldots, i_n}$.

We now show that $\alt(f_k) \ge 2^k-1$ by induction on $k$. For $k=1$, since $f$ depends only on $1$ variable the result holds. 

Suppose that the result holds for $f_k \in \calF$. 
For $f_{k+1} \in \calF$ on $n=2^{k+1}-1$ variables $\set{x_1,x_2,\ldots,x_n}$, let $T$ be the decision tree (as in~\cref{def:gap-fun}) 
computing $f_{k+1}$ of depth $k+1$ with $x_j$ being the root variable  for some $j \in [n]$. Let $T_1$ and 
$T_2$ be the left and right subtree of the root node $x_j$. Consider the Boolean 
function $f_1$ (resp. $f_2$) computed by the decision tree $T_1$ (resp. $T_2$). Note 
that since $T_1$ and $T_2$ are obtained from $T$ in this way, by~\cref{def:gap-fun}, 
$f_1$ and $f_2$ belongs to $\calF$ and computes the same function up to variable 
renaming. Also note that both $f_1$ and $f_2$ are on $m=2^k-1$ variables. Since both $T_1$ and $T_2$ are of depth $k$, by inductive hypothesis, $
\alt(f_1) \ge m$ and $\alt(f_2) \ge m$. Using this, we now construct a chain for $f_{k+1}$ of alternation $2^{k+1}-1$.

Let $\Omega_1 = var(f_1)$, $\Omega_2 = var(f_2)$. Let $\sigma_1 \in 
B_{\Omega_1}$ and $\sigma_2 \in B_{\Omega_2}$ be such that $\alt(f_1,\sigma_1) = 
2^k-1$ and $\alt(f_2,\sigma_2) = 2^k-1$. We now define a $\sigma \in B_{var(f)} = 
B_{\Omega_1 \cup \Omega_2 \cup \set{j}}$ as  
$\sigma(i) = \sigma_1(i)$  if  $ i \in \set{1,2,\ldots,m}$, 
$\sigma(m+1)  = j$ and $\sigma(m+1+i)  = \sigma_2(i)$  if $ i \in \set{1,2,\ldots,m}$.
By definition, $\sigma$ is indeed a bijection. The $\sigma$ obtained is pictorially represented in~\cref{fig:proof} for clarity of exposition.

We now claim that $\alt(f,\sigma) = 2^{k+1}-1$. To show this, consider the 
chain corresponding to $\sigma$ given by $0^n \prec y_1 \prec \ldots \prec y_m \prec y_{m+1} \prec 
\ldots \prec y_n = 1^n$ all belonging to $\zon$. By definition of $\sigma$, for $i \in 
[m]$ since $j^{th}$ bit of $y_i$  is $0$, $f(y_i) = f_1(y_i|
_{\Omega_1})$ and for $i \in \set{0} \cup [m]$, since $j^{th}$ bit of $y_{m+1+i}$ is $1$, $f(y_{m+1+i}) = f_2(y_{m
+1+i}|_{\Omega_2})$. Again by definition of $\sigma$, 
for $i \in [m]$, $0^m$ along with the elements $y_i|
_{\Omega_1}$ for $i \in [m]$ (in that order) is a chain 
witnessing $f_1$ alternating $m$ times and $y_{m
+1+i}|_{\Omega_2}$ for $i \in \set{0} \cup [m]$ (in that 
order) is a chain witnessing $f_2$ alternating $m$ 
times. Now observe that $f(y_m) \ne f(y_{m+1})$. This is 
because $f(y_m)$ is $f_1$ evaluated on $x_i = 1$ for all 
$i \in \Omega_1$ is the rightmost child of $T_1$ 
which is $1$ and $f(y_{m+1})$ is $f_2$ evaluated on $x_i 
= 0$ for all $i \in \Omega_2$  is the leftmost 
child of $T_2$ which is $0$. Hence $\alt(f, \sigma)  = 
\alt(f_1,\sigma_1) + \alt(f_2,\sigma_2) + 1 = 2m+1 = 
2^{k+1}-1$ completing the induction.  
\end{proof} 
In the next section, we show that for any Boolean function $f$, we can indeed upper bound 
alternation of $f$ by an exponential in decision tree depth of $f$.

\subsection{Alternation is at most Exponential in Decision Tree Depth }\label{sec:gap:ub}
In this section, we show that for all Boolean functions $f$, $\alt(f) \le 2^{\DT(f)+1}-1$. 
Markov~\cite{Mar58} studied a parameter 
closely related to $\alt(f)$ defined as \emph{decrease} (denoted by $\dc(f)$) where 
the definition is same as alternation except that the flips in the chain 
from $1$ to $0$ (corresponding to a decrease in the function value) alone are counted. Hence for any Boolean function $f$, $\alt(f) \in 
\set{2\dc(f)-1,2\dc(f) , 2\dc(f)+1}$. Markov~\cite{Mar58} showed the following tight connection between 
negations needed to compute a Boolean function and its decrease.

\begin{theorem}[Markov~\cite{Mar58}]\label{markov}
Let $\negs(f)$ be the minimum number of negations needed in any circuit computing $f$. Then, for any $f:\zon \to \zo$, $\negs(f) = \lceil \log (1 + \dc(f))\rceil$.
\end{theorem}
We also use the notion of a \emph{connector} for two Boolean functions which is a crucial idea used by Markov in proving the above theorem. A connector  of two Boolean functions $f_0$ and $f_1$ is a function $g(b_0,b_1,x)$ on $n+2$ bits such that $g(0,1,x) = f_0(x)$ and $g(1,0,x) = f_1(x)$. Markov showed the following remarkable bound that $\negs(g) \le \max\{\negs(f_0), \negs(f_1)\}$ (for a proof, see Claim 10.11 of Jukna~\cite{Juk12}). We now prove our result which follows from an inductive application of~\cref{markov}.

\begin{theorem}\label{alt:dt}
For any $f:\zon \to \zo$, $ \alt(f) \le 2^{\DT(f)+1}-1$.
\end{theorem}
\begin{proof}
Since $\alt(f) \le 2\dc(f) + 1$, it suffices to show that $\dc(f) \le 2^{\DT(f)}-1$. 
Proof is by induction on $\DT(f)$. 
For $\DT(f) = 1$, the function depends on at most $1$ variable, giving $\dc(f) \le 1 = 2^{\DT(f)}-1$. 
For any $f$ with $\DT(f) \ge 2$ computed by a decision tree $T$ of depth $k$, let $x_i$ be the variable 
queried at the root of $T$ for some $i \in [n]$. Define $f_0$ as $f$ restricted to $x_i=0$ and $f_1$ as 
$f$ restricted to $x_i=1$. Removing the node $x_i$ from $T$ gives two decision trees 
which computes $f_0$ and $f_1$ (respectively) giving $\DT(f_0) \le \DT(f)-1$ and $\DT(f_1) \le 
\DT(f)-1$. Hence by induction, $\dc(f_0) \le 2^{\DT(f_0)}-1 \le 2^{\DT(f)-1}-1$ and 
similarly, $\dc(f_1) \le 2^{\DT(f)-1}-1$.

Applying~\cref{markov}, we get, $\negs(f_0) = \lceil \log (\dc(f_0)+1)\rceil$ 
(and similarly for $f_1$). Let $g$ be 
the connector for $f_0$ and $f_1$. Since $f(x) = x_i \land f_1(x) \vee \neg x_i 
\land 
f_0(x)$, $f(x) = g(\neg x_i, x_i,x)$. Applying Markov's result on the number of 
negations needed in computing $g$, we get $\negs(f) \le \negs(g)+1 \le 1+\max
\{\negs(f_0), \negs(f_1)\}$. Since, $\max\{\negs(f_0), \negs(f_1)\} \le \left 
\lceil \log (2^{\DT(f)-1}-1+1 ) \right \rceil$ which is $\DT(f)-1$, $\negs(f) \le \DT(f)$. Applying~\cref{markov}, on $f$ completes the induction.  
\end{proof}
Note that for the family of functions $\calF$,~\cref{newgap:alt:dt} shows that the above result is asymptotically tight.
In~\cref{app:alt-dt}, we present a simpler proof of a slightly weaker result (in terms of an additive constant) that for all $f$, $\alt(f) \le 2^{\DT(f)+1}+1$. 

\begin{appsection}{Second Proof of Exponential Upper Bound on Alternation}{app:alt-dt}
Nechiporuk~\cite{Neci62}, also discovered independently by Morizumi~\cite{M09}, related the decrease of a function $f$ to the minimum number of negations needed in any formula computing $f$. 

\begin{theorem}[Nechiporuk, Morizumi] \label{negs:mori}
Let $\negs_F(f)$ be the minimum number of negations needed in any formula computing $f$. Then $\negs_F(f) = \dc(f)$.
\end{theorem}

We now give a weaker version (in terms of the additive constant) of~\cref{alt:dt}.
\begin{theorem}
For any Boolean function $f$, $\alt(f) \le 2^{\DT(f)+1}+1$
\end{theorem}
\begin{proof}
Since $\alt(f) \le 2\dc(f) + 1$, it suffices to show that $\dc(f) \le 2^{\DT(f)}$ to complete the proof. Let $T$ be a decision tree of depth $\DT(f)$ computing $f$. Given such a decision tree $T$, we can always obtain a formula computing $f$ with at most $2^{\DT(f)}$ leaves. 
Hence, $\negs_F(f)$ is upper bounded by $2^{\DT(f)}$ as all internal negations can be pushed to the leaves. Applying~\cref{negs:mori}, we get that $\dc(f) \le 2^{\DT(f)}$ which completes the proof.
\end{proof}
\end{appsection}
\subsection{Super Linear Gaps Between Alternation and Sensitivity}\label{sec:compose}

In this section, we exhibit a family of functions with a super-linear gap between 
alternation and sensitivity with high decision tree depth. We start by giving a motivation for this study.

In~\cref{sec:gap:lb}, we showed the existence of a family $\calF$ of Boolean functions with alternation at least exponential in sensitivity and alternation is at least logarithm of sparsity. Hence this family $\calF$ rules out the possibility of 
upper bounding alternation by a polynomial sensitivity or a polynomial in logarithm of sparsity for all Boolean functions which would have settled the \textsc{Sensitivity Conjecture} and \textsc{\XOR~Log-Rank Conjecture} by the results of Lin and Zhang~\cite{LZ17}. However, the above mentioned conjectures holds true for the family $\calF$ since they have shallow decision trees computing them.

\noindent
\textbf{\textsc{Sensitivity Conjecture} is true for all $f_k \in \calF$ : } To argue this, we first observe that as $f_k$ depends on all its inputs, $s(f_k) = \Omega(\log n_k)$ \cite{Sim83}. Hence we can conclude that $\DT(f_k) = O(\sens(f_k))$. Since $\deg(f_k) \le \DT(f_k)$ (see~\cite{BW02}), we get $\deg(f_k) = O(\sens(f_k))$. Note that, this argument is valid for any function $f$ that depends on $n$ variables, and $\DT(f) = O(\log n)$.

\noindent
{\bf \textsc{\XOR~Log-Rank Conjecture} is true for all $f_k \in \calF$ : }Let $T_k$ be the decision tree computing $f_k$ (as per~\cref{def:gap-fun}). Consider the Fourier expansion of $f_k:\zo^{n_k} \to \pmo$ given by,
\begin{align*}
f_k(x) & =  \sum_{a \in \zo^k} val(\ell_a) \prod_{i_j \in var(\ell_a)} \frac{1+(-1)^{x_{i_j}}(-1)^{a_j}}{2} \\
& = \frac{-1}{2^k} \sum_{a \in \zo^k} val(\ell_a) \sum_{S \sse var(\ell_a)} \chi_S(x)\prod_{j : i_j\in S} (-1)^{a_j}
\end{align*}
where, for the leaf $\ell_a$ of $T_k$, $a\in \zo^k$ is the value taken by the variables with indices $var(\ell_a)= \set{i_1, i_2, \ldots, i_k }$ 
in the path to $\ell_a$ and $val(\ell_a) \in \pmo$ is the value of $f_k$ on the input $a$. For $a,b \in \zo^k$, if $a \ne b$ then $var(\ell_a) \ne 
var(\ell_b)$ as the variable set in each path of $T_k$ are different. Hence  $\chi_S(x)$ for $S=var(\ell_a)$ is never canceled in the above expansion thereby concluding that $\sps(f_k) \ge 2^k$. Since $\DT(f_k) 
= k$, we have $\DT(f_k) \le \log \sps(f_k)$. Using the fact that for any Boolean function 
$f$, $\CC_\oplus(f) \le 2{\sf DT}(f)$ \cite{MO09}, $\CC_\oplus (f_k) = O(\log \sps(f_k))$ 
thereby concluding that \textsc{\XOR~Log-Rank Conjecture} is true for  all $f_k \in \calF$.

As stated before, the key property of $ f \in \calF$ due to which both the conjectures are true is that the decision tree complexity of $f \in \calF$ is logarithmic in the number of variables of $f$. In fact for any $f:\zon \to \zo$ which satisfies $\alt(f) \ge 2^{\DT(f)}$ must have $ \DT(f) = O(\log n)$ as $\alt(f) \le n$.  Hence, we ask if there exists a family of functions $f$ where $\alt(f)$ grows faster than $\sens(f)$ but $\DT(f)$ is not very small. 

\noindent
{\bf Super-linear Gap between Alternation and Sensitivity :} 

In the rest of this section, we answer this question by exhibiting a  family of Boolean functions $\calG = \set{g_k:\zo^{n_k} \to \zo \mid k \in \N}$ with $\alt(g_k) = \omega(\sens(g_k))$ and $\DT(g_k)$ is $\Omega(n_k^{\log_6 3})$. 
Before proceeding, we show a lower bound on the alternation of composition of two Boolean functions in terms of its alternation.

\begin{lemma}\label{alt:composition}
For any $g:\zon \to \zo$, with $g(0^n) \ne g(1^n)$ and any $f:\zo^m \to \zo$,
$ \alt(f \circ g) \ge \alt(f)\cdot \alt(g)$.
\end{lemma}
\begin{proof}
Without loss of generality, assume $g(0^n) = 0$ and $g(1^n) = 1$ 
(otherwise work with $\neg g$ as $\alt(g) = \alt(\neg g)$). 
Let $A = (0^n = z_0 \prec z_1 \prec \dots \prec z_n= 1^n)$ be a chain on $\zon$ such that the alternation 
of $g$ is maximum. Consider any maximum alternation chain of $f$ and let $\sigma$ be the 
permutation associated with the chain. We exhibit a chain $B= (y_0, y_1, \ldots, y_{nm})$ on 
$\zo^{nm}$ with $\alt(f)\cdot \alt(g)$ many alternations. The idea is to glue $m$ copies of $A$ which are appropriately separated in a way such that for every flip of $f$, $f \circ g$ flips $\alt(g)$ times in the chain $B$. 

We divide inputs in the chain $B$ into $m$ blocks of size $n$ each. We say that for a $k \in [nm]$, the input $y_k \in \zo^{nm}$ 
belongs to the \emph{block} $b(k)$ if $b(k) = \lceil \frac{k}{n} \rceil$. We define the \emph{position} of $y_k$ in its block, $pos(k)$, as $n$ if $n \mid k$ and  $(k \mod n)$ otherwise. Hence $k = (b(k)-1) \cdot n + pos(k)$. Let 
$y_k = (x_1^k,x_2^k,\dots,x_m^k)$ where $x_i^k \in \zon$. For $k=0$, define $x_i^k = 0^n$ for all $i 
\in [m]$ and for $k=nm$, define $x_i^k = 1^n$ for all $i \in [m]$. For the remaining values of $k$, 
$x_i^k$ for $i \in [m]$ is defined  as
\[ x_i^k = \begin{cases}
		  z_n= 1^n & \text{ if } i \in \set{\sigma(1), \sigma(2), \ldots, \sigma(b(k)-1)} \text{ and } b(k) \ge 2 \\
		  z_{pos(k)} & \text{ if } i = \sigma(b(k)) \\
		  z_0 =0^n  & \text{ otherwise } 
		\end{cases}
\]
We can see that $y_0 = 0^{nm}$ and $y_{nm} = 1^{nm}$ and for $k \in [nm]$, from the above 
definition, $y_{k-1} \prec y_{k}$ as $\forall i \in [m] ~x_i^{k-1} \prec x_i^{k}$. 
We now argue that $f\circ g$ alternates at least $\alt(f) \cdot \alt(g)$ times in the chain 
$B$. Consider the input $(g(x_1^k), g(x_2^k),\ldots, g(x_m^k))$ to the function $f$ and let 
$y_k$ belong to the block $b(k)$.

Consider the case when $b(k)=1$. In this case, all of $x_i^k$ except $i=\sigma(b(k))$ is $0^n$. As 
long as $y_k$ stays within the block $b(k)$, the input the function $f$ changes its value only 
at $x_{\sigma(b(k))}^k = x_{\sigma(1)}^k$. Since $g(x_{\sigma(b(k))}^k)$ changes its value $
\alt(g)$ times, $f\circ g$ will also alternate $\alt(g)$ times if value of $f$ changes on 
flipping location $\sigma(b)$ in its input. 

For $k$ such that $b(k) > 1$, by definition of $y_k$, $x_{\sigma(1)}^k, \ldots, x_{\sigma(b(k)-1)}
^k$ is $1^n$ and $x_{\sigma(b(k)+1)}^k \ldots, x_{\sigma(m)}^k$ is $0^n$. Since $g(0^n) = 0$ and 
$g(1^n) = 1$, the input to $f$ will be either $r_1 = \lor_{i=1}^{b(k)} e_{\sigma(i)}$ or $r_0 = 
\lor_{i=1}^{b(k)-1} e_{\sigma(i)}$. Since $g$ alternates $\alt(g)$ times thereby changing the 
input to $f$ between $r_0$ and $r_1$, $f\circ g$ will also alternate $\alt(g)$ times if value 
of $f$ changes on flipping location $\sigma(b(k))$. 

Thus in both cases, if $f$ alternates once, $f\circ g$ alternates $\alt(g)$ in the chain $B$. 
Since $f$ alternates $\alt(f)$ times on $\sigma$, $f\circ g$ alternates $\alt(f) \cdot 
\alt(g)$ times in the chain $B$.  
\end{proof}
For $f = \lor_n$ and $g$ being a  parity on $m$ bits for any odd integer $m$, $\alt(f \circ g) \ge mn$ by~\cref{alt:composition} while $\alt(f)\cdot \alt(g) = m$. Thus, in general, it is not true that $\alt(f\circ g) \le \alt(f) \cdot \alt(g)$ and hence \cref{alt:composition} is not tight. Using~\cref{alt:composition}, we prove the following Corollary. 
\begin{corollary}\label{corr:alt-compose}
For any $h:\zon \to \zo$, with $h(0^n) \ne h(1^n)$, for any $k \ge 2$, $\alt(h^{\circ k}) \ge \alt(h)^k$.
\end{corollary}
\begin{proof}
By induction on $k$. For $k=2$, 
applying~\cref{alt:composition} with $f = g = h$, we get that $\alt(h \circ h) \ge 
\alt(h)^2$. Now for $k > 2$, applying~\cref{alt:composition} with $f = h^{\circ k-1}$ and 
$g=h$ and by inductive hypothesis, $\alt(h^{\circ k-1} \circ h) \ge \alt(h^{\circ k-1}) \cdot 
\alt(h) \ge \alt(h)^{k-1} \cdot \alt(h)$. Hence $\alt(h^{\circ k}) = \alt(h^{\circ k-1} \circ 
h) \ge \alt(h)^k$.  
\end{proof}

We use~\cref{corr:alt-compose} to exhibit a family of Boolean functions for which alternation is super-linear in  sensitivity. 

\introthm{alt:sens:sep}{
There exists a family of Boolean functions $\calG =\{ g_k :\zo^{n_k} \to \zo$ $\mid k \in \N\}$ such that $\alt(g_k) \ge \sens(g_k)^{\log_3 5}$ while $\DT(g_k)$ is  $\Omega(n_k^{\log_6 3})$. 
} 
\begin{proof}
Consider the address function 
$\ADDR_t:\zo^{t+2^t} \to \zo$ defined as 
$\ADDR_t(x_1,x_2,\dots,x_t,\allowbreak y_0y_1,y_2,\dots,y_{2^t-1}) = y_{int(x_1x_2\dots x_t)}$ where $int(x)$ is the integer corresponding to the binary string $x$. 
Consider the chain $
(000000,  001000,101000 , \allowbreak 101010 , 111010,111011,111111)$. Since, $\ADDR_2$ changes value $5$ times along this chain, $\alt(\ADDR_2) \ge 5$ while $\sens(\ADDR_2) = 3$. We consider the family of functions $\calG = \set{g_k\mid k \in N}$ obtained by 
composing $\ADDR_2$ $k$ times. Since sensitivity of composed function is 
at most the product of their sensitivity~\cite{Tal13}, $\sens(g_k) \le \sens(\ADDR_2)^k = 3^k$. Since $g_1=\ADDR_2$ is $0$ on 
all zero input and $1$ on all ones input,  applying 
\cref{corr:alt-compose}, $\alt(g_k) \ge 5^k \ge \sens(g_k)^{\log_3 5}$. Note that $\DT(g_k) = \DT(\ADDR_2)^k$ (as decision tree depth multiplies under composition~\cite{Tal13}). Hence $\DT(g_k) = 3^k$ which is $n_k^{\log_6 3}$ where $n_k$ is the number of variables of $g_k$ and hence does not grow logarithmic in $n_k$.  
\end{proof}

\begin{remark}
We observe that for the family of functions $\calG = \set{g_k : \zo^{n_k} \to \zo \mid k \in \N}$ of~\cref{alt:sens:sep}, $\log \sps(g_k) \le 2\deg(g_k) = 2 \cdot 3^k$ (see Corollary 2.8, \cite{T14} for a proof of the inequality).  In~\cref{alt:sens:sep}, we also showed that $\alt(g_k) \ge 5^k$ which is at least $(0.5\log \sps(g_k))^{\log_3 5}$.
Hence the same family also exhibits a super-linear gap between $\alt(g_k)$ and $\log \sps(g_k)$.  
\end{remark}

\section{\textsc{\XOR~Log-Rank Conjecture} for Bounded Alternation Boolean Functions}
\label{sec:xor-polylogn}
In this section, we prove the \textsc{\XOR~Log-Rank Conjecture} for $f$ when $
\alt(f)$ is at most $\poly(\log n)$.  Before proceeding, we prove that for all 
Boolean functions $f$, $\deg(f) \le \alt(f) \degtwo(f)  \degm(f)$  
(\cref{eq:deg:deg2} from Introduction) in~\cref{deg:deg2:lem}. We remark that the case $m=2$ is already observed by Lin and Zhang (Theorem 14,~\cite{LZ17}). 

\begin{lemma}\label{deg:deg2:lem}
For any Boolean function $f:\zon \to \zo$, and $m > 1$,  \[\deg(f) \le \alt(f) \cdot \degtwo(f) \cdot \degm(f) \tag{1}\] 
\end{lemma}
\begin{proof}
The proof of this lemma closely follows the argument of Lin and Zhang (Theorem 14,
\cite{LZ17}).
Fix any Boolean function $f$.  
Buhrman and de Wolf~\cite{BW02} showed that given a certificate of $f$ on $0^n$ of 
size $\cert(f,0^n)$ and a polynomial representation of $f$ over reals with degree of 
 $\deg(f)$, $\DT(f) \le \cert(f,0^n)\cdot \deg(f)$. Their idea is that any 
certificate of $f$ must set at least one variable in all monomials of maximum degree 
in the polynomial representation of $f$. Hence, querying variables in a certificate 
must reduce the degree of the function by at least once. Observe that the same 
argument holds even if the polynomials are represented over $\Z_m$. Hence $\DT(f) \le 
\cert(f,0^n) \cdot \degm(f)$. 

Applying a result of 
Lin and Zhang, who showed that $\cert(f,0^n) \le \alt(f) \degtwo(f)$ (Lemma 12(2) of~
\cite{LZ17}), we get that, $\DT(f) \le \alt(f) \degtwo(f)\degm(f)$. Observing that $
\deg(f) \le \DT(f)$ (see~\cite{BW02}) completes the argument. 
\end{proof}
Applying~\cref{deg:deg2:lem} with $m=2$, we have for any Boolean function $f$,

\begin{equation} \label{alt:deg:deg2}
\deg(f) \le \alt(f) \cdot \degtwo(f)^2  
\end{equation}

As mentioned above, bound \cref{alt:deg:deg2} on $\deg$ is implicit in the result of~\cite{LZ17}. Recently, Li and Sun showed that $\deg(f) = O(\alt(f) \deg_m(f)^2)$ (Theorem 2,~\cite{LS18}) implying that \cref{alt:deg:deg2} holds not just for $m=2$ but for all integers $m$. 

We now proceed to prove the main result of this section.
As a first step towards showing $\alt(f) \le \poly(\log n)$ implies that the \textsc{\XOR~Log-Rank Conjecture} holds for $f$, we prove the following bound on the weighted average of the Fourier coefficients, weighted by the number of elements.

\begin{proposition}
 \label{fourier:lb}
For an $f:\zon \to \pmo$ that depends on all its inputs, $ \sum_S |\fhat(S)||S| \ge n $.
\end{proposition} 
\begin{proof}
It suffices to show that for every $i \in [n]$, $\sum_{S:i \in S} |\fhat(S)| 
\ge 1$. Fix an $i \in [n]$. Since \[f(x) = \sum_S \fhat(S)\chi_S(x) = \sum_{S \subseteq [n] \setminus \{i\}} (\fhat(S) 
+ \fhat(S \cup \{i\})(-1)^{x_i}) \prod_{j \in S}(-1)^{x_j}\] for any $S \subseteq [n]\setminus \{i\}$ and  $b \in \zo$, $\widehat{f|_{x_i=b}}(S) = \fhat(S) + (-1)^b\cdot \fhat(S \cup \{i\})$.
Hence we conclude that for any $x$,
$ f_{x_i=0}(x)-f_{x_i=1}(x) = \sum_{S \subseteq [n] \setminus \{i\}} 2\fhat(S \cup \{i\}) \prod_{j \in S}(-1)^{x_j}$.  Now taking absolute values on both sides and applying triangle inequality,
\begin{equation*}
\left |\frac{f_{x_i=0}-f_{x_i=1}}{2} \right |  = \left |\sum_{S \subseteq [n] \setminus \{i\}} \fhat(S \cup \{i\}) \prod_{j \in S}(-1)^{x_j} \right | 
 \le \sum_{S \subseteq [n] \setminus \{i\}} \left |  \fhat(S \cup \{i\}) \right | \label{eq:fourier:lb}
\end{equation*}
Since $f$ is sensitive at $i$ on some input $a \in \zon$, 
for the input $a'$ obtained by removing $i^{th}$ bit from $a$, $|f_{x_i=0}(a')-f_{x_i=1}(a')| = 2$ implying $\sum_{S : i \in S} |\fhat(S)| \ge 1$ by the above equation which completes the proof. 
\end{proof}

We show that for Boolean functions $f$  with range as $\zo$, if $\alt(f) \le \poly(\log n)$, then $\deg(f) \le \poly(\log \sps(f))$. This implies that the \textsc{\XOR~Log-Rank Conjecture} holds\footnote{Using the facts that $\CC_\oplus(f) \le 2{\sf DT}(f)$ \cite{MO09} and $\DT(f) \le \deg(f)^4$	\cite{BW02}, $\CC_\oplus(f) =
	O(\deg(f)^4)$ implying $\CC_\oplus (f) = O(\poly(\log \sps(f))$.} for $f$.
 As a first step, using~\cref{fourier:lb}, we show that if $\deg(f) \le \poly(\log n)$, then 
 $\deg(f) \le \poly(\log \sps(f))$ (\cref{deg:sps}).  We then argue using~
 \cref{alt:deg:deg2} that $\alt(f) \le \poly(\log n)$ implies that $\deg(f) \le \poly(\log n)$ 
 proving~\cref{alt:polylog}.

\begin{lemma} \label{deg:sps}
For an $f:\zon \to \zo$ which depends on all its inputs and for large enough $n$, 
if $\deg(f) \le (\log n)^c$ for some $c > 0$, then $\deg(f) \le (\log 
\sps(f))^c$.
\end{lemma}
\begin{proof}
Since the $f$ depends on all the inputs, applying~\cref{fourier:lb} to $g(x) = 1-2f(x)$,
$n \le \sum_S |\ghat(S)||S| \allowbreak \le $ $ \deg(f)\sum_S |\ghat(S)| \le \deg(f)\sqrt{\sps(f)}$. In concluding this, we used the fact that maximum sized index $|S|$ for which $\ghat(S) \ne 0$ is $\deg(f)$ and $\sum_S |\ghat(S)| \le \sqrt{\sps(f)}$ (by Cauchy-Schwartz and $\sps(f) = \sps(g)$).
Thus, $\sqrt{\sps(f)}\cdot \deg(f) \ge n$. Since $\deg(f) 
\le (\log n)^c$, we have $\sqrt{\sps(f)} \ge \frac{n}{(\log n)^c} \ge \sqrt{n}$ for 
large enough $n$. Hence $\deg(f) \le (\log n)^c \le  (\log \sps(f))^c$. 
\end{proof}
\introthm{alt:polylog}{
For large enough $n$, the \textsc{\XOR~Log-Rank Conjecture} is true for all $f:\zon \to \zo$, such that  $\alt(f) \le \poly(\log n)$ where $f$ depends on all its inputs.} 
\begin{proof}
 If $\degtwo(f) =1$, then $f$ is a parity function and the \textsc{\XOR~Log-Rank Conjecture} holds for $f$. Hence we can assume that $\degtwo(f) > 1$. If $\alt(f) \le \degtwo(f)$, then by \cref{alt:deg:deg2}, we have $\deg(f) \le 
\degtwo(f)^3 \le \log \sps(f)^3$ (since $\degtwo(f)> 1$, $\degtwo(f) \le \log \sps(f)$~\cite{BC99}). Hence the \textsc{\XOR~Log-Rank Conjecture} holds for $f$. If $\alt(f) > 
\degtwo(f)$, then by \cref{alt:deg:deg2}, $\deg(f) < \alt(f)^3$. Since $\alt(f) \le 
\poly(\log n)$, we have $\deg(f) \le \poly(\log n)$. Applying 
\cref{deg:sps}, we get that $\deg(f) \le \poly(\log \sps(f))$.  
\end{proof}

\begin{remark}
It should be noted that for $f$ satisfying conditions of~\cref{alt:polylog}, the \textsc{Sensitivity Conjecture} is true. This is because for $f$ that depends on all its inputs, $\sens(f) = \Omega(\log n)$~\cite{Sim83} implying that $\alt(f) \le \poly(\log n) \le \poly(\sens(f))$. Hence the \textsc{Sensitivity Conjecture} is true for $f$ by the result of Lin and Zhang~\cite{LZ17}.
\end{remark}

We conclude this section, by giving an alternate proof for the result of Lin and Zhang~\cite{LZ17} that if for all Boolean functions $f$, $\alt(f) \le \poly(\log \sps(f))$, then the \textsc{\XOR~Log-Rank Conjecture} is true by making use of  the bound  (\ref{alt:deg:deg2}) on $\deg(f)$. Then, we give a trade-off between sensitivity and sparsity of Boolean functions.

\begin{theorem}
For any Boolean function $f$ if $\alt(f) \le (\log \sps(f))^c$ then \textsc{\XOR~Log-Rank Conjecture} holds.
\end{theorem}
\begin{proof}
For any Boolean function $f$ with $\alt(f) \le (\log \sps(f))^c$ we have the following exhaustive cases.
\begin{description}
	\item[[Case 1 : $\alt(f) = 1$]] : When $\alt(f) = 1$, $f$ is unate. Hence 
	the \textsc{\XOR~Log-Rank Conjecture} holds due to the result of~\cite{MO09}.
	\item[[Case 2: $\degtwo(f)$ is constant ]] :	When $\degtwo(f)$ is constant, the \textsc{\XOR~Log-Rank Conjecture} holds due to the result of~\cite{TWXZ13}.
	\item[[Case 3 : $1 < \alt(f) < \degtwo(f)$]] :  For this case, when
		$\alt(f) < \degtwo(f)$, by \cref{alt:deg:deg2}, $\deg(f) \le \degtwo(f)^3$ 
		which is upper bounded by $(\log \sps(f))^3$ (since $\degtwo(f)> 1$, $\degtwo(f) \le \log \sps(f)$~\cite{BC99}). Thus $\deg(f) \le  (\log \sps(f))^3$. On the other hand, $\CC_
		\oplus(f) \le 2{\sf DT}(f)$ \cite{MO09}. Also it is known that $\DT(f) \le \deg(f)^4$
	\cite{BW02}. Hence $\CC_\oplus(f) =
	O(\deg(f)^4)$ giving $\CC_\oplus (f) = O((\log \sps(f))^{12})$. 
	
	\item[[Case 4 : $\degtwo(f) \le \alt(f) \le (\log \sps(f))^c$]] : we have
	$\CC_\oplus(f)= O(\deg(f)^{\ell}) \le O(\alt(f)^{3\ell})$ where the last inequality follows by 	\cref{alt:deg:deg2} in this case. By assumption, $\log \sps(f) \ge
	alt(f)^{1/c}$. Thus $\CC_\oplus(f) = O(\log
	\sps(f))^{3\ell/c})$.
\end{description} 
\end{proof}

From~\cref{fourier:lb}, we derive a trade-off between sensitivity and sparsity of Boolean functions.
\begin{corollary}
\label{alt:rem} For any Boolean function $f$ which depends on all its $n$ bits,
$\sens(f)\allowbreak \sqrt{\sps(f)} \ge n$. 
\end{corollary}
\begin{proof}
 Applying Cauchy-Schwartz inequality to  $\sum_S |\fhat(S)||S| \ge n $ from~\cref{fourier:lb}, we get 
		\begin{align*}
		\sum_S |\fhat(S)| |S| & \le \sqrt{\sum_S |S|^2 \fhat(S)^2
		\cdot \sps(f)} \le   \sens(f)	\sqrt{\sps(f)} \label{alt:sens:eq} 
		\end{align*}	

The last inequality follows since $\sum_S |S|^2 \fhat(S)^2 = \frac{1}{2^n}\sum_x \sens(f,x)^2 \le \sens(f)^2$ (see Exercise 2.20~\cite{OD14}).
\end{proof}

\section{Three Further Applications of the $\deg$ vs $\degtwo$ Relation}\label{sec:deg-deg2:app}
We showed that for all Boolean functions $f$, $\deg(f) \le \alt(f) 
\degtwo(f)^2$~(\cref{alt:deg:deg2}) in~\cref{sec:xor-polylogn}. We now give three 
applications of this result. Firstly, we show that Boolean functions of bounded 
alternation must have $\F_2$ degree $\Omega(\log n)$. Secondly, we partially answer a 
question raised by Kulkarni and Santha~\cite{KS13} on the sparsity of monotone 
Boolean functions by showing a variant of their statement. Thirdly, we observe that~
\cref{alt:deg:deg2} improves a bound on $\I[f]$ due to Guo and Komargodski~\cite{GK17}.
\paragraph{$\F_2$-degree of Bounded Alternation Functions :} For an $n$ bit monotone Boolean function $f$, we have $\deg(f) \le \degtwo(f)^2$. If $f$ depends on all its $n$ bits then $\deg(f) \ge \log n - O(\log \log n)$~\cite{NS92}. This implies that for monotone functions, $\degtwo(f) \ge \Omega(\sqrt{\log n})$. We now present a short argument improving this bound using~\cref{alt:deg:deg2}.

Suppose $f:\zon \to \zo$ and it depends on all its input bits. 
Gopalan \etal~\cite{GLS09} showed that for any such Boolean function, $\deg(f) \ge n/2^{\degtwo(f)}$. Along with~\cref{alt:deg:deg2}, this implies that for an $f$  with $\alt(f) \le n^{\epsilon}$ where $0 < \epsilon < 1$, 
$\degtwo(f)^2 2^{\degtwo(f)} \ge n^{1-\epsilon}$. Hence, $\degtwo(f) \ge (1-\epsilon)\log n - 2(\log \left ( \frac{1-\epsilon}{2} \log n\right))$. This gives us the following corollary.
\begin{corollary}\label{corr:deg2lb}
Fix any $0 <\epsilon < 1$. Let $f:\zon \to \zo$ be such that $\alt(f) \le n^{\epsilon}$ and it depends on all its inputs. Then, 
\[\degtwo(f) \ge (1-\epsilon)\log n - O\left (\log \left ( \frac{1-\epsilon}{2} \log n\right)\right).\] 
\end{corollary}
Hence, Boolean functions whose alternation is at most $n^{\epsilon}$, cannot have a constant $\F_2$-degree. However, this need not be the case if we allow $\alt(f)$ to be $n$ (for instance, parity on $n$ bits).

\paragraph{Fourier Spectrum for Bounded Alternation Functions :} 
Kulkarni and Santha \cite{KS13} studied certain special Boolean functions which are indicator functions $f_{\calM}$ of a bridgeless matroids $\calM$ on ground set $[n]$. While it is known that for any $f$, $\log \sps(f) \le 2\cdot \deg(f)$ (Exercise 1.11,  \cite{OD14}), Kulkarni and Santha showed that this upper bound is asymptotically tight for $f = f_{\calM}$. They observed that $f_{\calM}$ is a monotone function (by virtue of the underlying support set being a matroid) and asked if a similar statement holds for the general class of monotone Boolean functions. More precisely, they asked whether $\log \sps(f) = \Omega(\deg(f))$ for every monotone Boolean 
function $f$.

We show that for functions of constant alternation (which includes monotone functions), $\log \sps(f)$ is relatively large. This can be seen as a variant of the question posed by Kulkarni and Santha.

\introthm{alt:bound}{
For Boolean functions $f$ with $\alt(f) = O(1)$, $\log \sps(f) =\Omega(\sqrt{\deg(f)})$.} 
\begin{proof}
Observe that for $\degtwo(f)=1$, $f$ is parity of constant number of variables or its negation as $\alt(f) = O(1)$. Hence the result holds for this case. For $\degtwo(f) > 1$,  by~\cref{alt:deg:deg2} and the result that $\degtwo(f) \le \log \sps(f)$ when $\degtwo(f) > 1$~\cite{BC99},  we get that $\deg(f) = O(\degtwo(f)^2) = O(\log \sps(f)^2)$.
  \end{proof}

Note that~\cref{alt:bound} does show that logarithm of sparsity of monotone functions is nearly close to the upper bound possible but does not completely settles the question of Kulkarni and Santha.

\paragraph{Improved Upper Bound for $\I[f]$ :} 
For an $n$ bit Boolean function, the best known upper bound of $\I[f]$ in terms of 
$\alt(f)$ is $\I[f] \le O(\alt(f) \sqrt{n})$ due to Guo and Komargodski~\cite{GK17} using a probabilistic argument. In~\cref{app:inf-alt} we give a simpler proof using a recent characterization of alternation due to Blais \etal~\cite{BCOST15}. 
Since $\I[f] \le \deg(f)$ (Fact 3.7, \cite{OD14}), the bound (\ref{alt:deg:deg2}) on $\deg(f)$ gives an improvement over the known bound on $\I[f]$ when $\degtwo(f) < \sqrt[4]{n}$.

\begin{appsection}{Simpler Proof of $\I[f] \le \alt(f) \sqrt{n}$}{app:inf-alt}
Suppose $f:\zon \to \zo$ is   a parity of $k$ monotone 
functions. Hence, in any chain from $0^n$ to $1^n$ in $\calB_n$, $f$ changes 
its value at most $k$ times. Blais \etal~\cite{BCOST15} recently 
showed that the converse of the statement is also true.
\begin{proposition}[Characterization of Alternation~\cite{BCOST15}] 
	\label{thm:alt-char}
	Let $f:\zon \to \zo$. Then
	there exists $k = \alt(f)$ monotone functions $g_1, \dots, g_k$ each from
	$\{0,1\}^n$ to $\{0,1\}$ such that
	\[f(x) = \begin{cases}
		 \oplus_{i=1}^k g_i & \text{ if } f(0^n) = 0 \\
		 \neg \oplus_{i=1}^k g_i & \text{ if } f(0^n) = 1
		\end{cases}
	\]
\end{proposition}

Let $f$ be such that $\I[f] = \Omega(n)$. 
Observe (see Fact 2.14,~\cite{OD14}) that an $\Omega(n)$ lower bound on influence implies that a 
\emph{constant} fraction of edges in the Boolean hypercube must have $f$ evaluating 
to different values on their end points. Blais \etal~\cite{BCOST15} used this idea 
and a probabilistic argument to prove that the expected alternation in a random chain 
of the hypercube must be large -- inferring that $\alt(f) = \Omega(\sqrt{n})$. Later 
Guo and Komargodski~\cite{GK17}, again using a similar probabilistic argument, 
generalized it to show that $\I(f) = O(\alt(f) \sqrt{n})$. 

We now give a simpler proof of the statement $\I[f] \le O(\alt(f) \sqrt{n})$.
\begin{lemma}
For any $n$ variable Boolean function $f$, $\I[f] \le \alt(f) \sqrt{n}$.
\end{lemma}
\begin{proof}
For any Boolean function $f_1, f_2$ on $n$ bits and for any $x \in \zon$, it is easy to see that $\sens(f_1 \oplus f_2 ,x) \le \sens(f_1,x) + \sens(f_2,x)$ (see Lemma 28 of~\cite{ST16}). Hence, 
\begin{equation} \label{eq:inf:ub}
\I[f_1 \oplus f_2] = \Expt_x(\sens(f_1 \oplus f_2,x)) \le \Expt_x (\sens(f_1,x)) + \Expt_x(\sens(f_2,x)) = \I[f_1] + \I[f_2]
\end{equation}

Suppose $f(0^n) = 0$ and $\alt(f) = k$. Applying~\cref{thm:alt-char}, we get that there exists $k$ monotone functions, $g_1, g_2, \ldots, g_k$ all on $n$ variables such that $f(x) = g_1(x) \oplus g_2(x) \ldots \oplus g_k(x)$. Repeatedly applying~\cref{eq:inf:ub} on $f$, we get that $\I[f] \le \sum_{i=1}^k \I[g_i]$. Since $g_i$s are monotone, $\I[g_i] \le \sqrt{n}$ (Exercise 2.23, \cite{OD14}). Hence $\I[f] \le k \sqrt{n} = \alt(f)\sqrt{n}$. For the case when $f(0^n) = 1$   in the decomposition of~\cref{thm:alt-char} all $g_i$s will still be monotone except one which will be negation of a monotone function. But even for such functions influence is upper bounded by $\sqrt{n}$ (Exercise 2.5(b), \cite{OD14}). Hence a similar argument holds in this case too.  
\end{proof} 
\end{appsection}

\begin{proposition}\label{inf:alt}
For any $f:\zon \to \zo$, $\I[f] \le \alt(f) \cdot \degtwo(f)^2$.
\end{proposition}
\noindent
This immediately gives improved learning algorithms for functions of bounded alternation in the PAC learning model.

Blais \etal~\cite{BCOST15} gave a uniform learning algorithm for the class of functions $C_t$ computable by circuits with at most $t$ negations that can learn an $f \in C_t$ from random examples with error $\epsilon > 0$ in time $n^{O(2^t\sqrt{n})/\epsilon}$ where $t \le O(\log n)$. In terms of alternation, the runtime is $n^{O(\alt(f)\sqrt{n})/\epsilon}$.
The main tool used in this area is the low degree learning algorithm due to Linial, Mansour and Nisan \cite{LMN93} using which the following result is derived in~\cite{OD14}.
\begin{lemma}[Corollary 3.22 and Theorem 3.36, \cite{OD14}]\label{learning}
For $t \ge 1$, let \[\calA_t = \set{f \mid  f:\pmon \to  \pmo, \I[f]  \le t}\]  and 
\[\calB_t = \{f \mid f:\pmon \to \pmo, \allowbreak \deg(f) \le t\}\] Then $A_t$ can be 
learned from random examples with error $\epsilon$ in time $n^{O(t/\epsilon)}$ for any $\epsilon \in (0,1)$
and $B_t$ can be exactly learned from random examples in time 
$n^t\poly(n , 2^t)$.
\end{lemma}
The claimed result follows from~\cref{inf:alt}. Applying~\cref{learning}, we obtain 
\begin{itemize}
\item  an exact learning algorithm from random examples with a runtime of 
$n^{O(k)} \poly(n,2^{O(k)})$ for $k=\alt(f)\degtwo(f)^2$ and
\item an $\epsilon$ error learning algorithm from random examples with a runtime $n^{O(\alt(f)\degtwo(f)^2/\epsilon)}$ 
\end{itemize}
thereby removing the dependence on the parameter $n$ in the exponent and improving the runtime for those $f$ such that $\degtwo(f) < \sqrt[4]{n}$. 

\section{Discussion and Open Problems}
In this paper, we showed a limitation of alternation as a Boolean function 
parameter  in settling the \textsc{Sensitivity Conjecture} and \textsc{\XOR~Log-
Rank Conjecture}. In spite of this limitation, we derived that both the above 
conjectures are true for functions whose alternation is upper bounded by $
\poly(\log n)$. En route the proof, we showed that the degree can be upper bounded 
in terms of $\F_2$ degree and alternation (\cref{alt:deg:deg2}) and demonstrated 
its use with three applications. 
In conclusion, we propose the following three directions for further exploration.
\begin{description}
\item[Parameter Trade-offs :]
The family of Boolean functions $\calF$ (in \cref{def:gap-fun}) have the drawback 
that their decision tree depth is very small while exhibiting an exponential gap 
between alternation and sensitivity, and alternation and logarithm of sparsity. On 
the other hand, the family of Boolean functions $\calG$ obtained (in
\cref{alt:sens:sep}) have a large decision tree depth  but could achieve the same 
with only a super linear gap between the above mentioned parameters. These two 
family of functions seems to be at the two extremes ends in terms of the gap 
achievable and the decision tree depth. Thus, an open problem would be to show a 
trade-off between the lower bound on the decision tree depth and the gap that 
can be proven for these parameters. 
\item[Monotone functions have dense spectrum :] Can we show that for every monotone function $f$, $\log \sps(f) = \Omega(\deg(f))$. Note that~\cref{alt:bound} shows that $\log \sps(f) =\Omega(\sqrt{\deg(f)})$  and hence only partially settle this question of Kulkarni and Santha~\cite{KS13}.
\item[Improving upper bound on $\deg(f)$ :] \cref{eq:deg:deg2} says that for any 
Boolean function $f$ and $m \ge 2$, $\deg(f) \le \alt(f) \cdot \degtwo(f) \cdot 
\degm(f)$. Note that the upper bound is weak if $\degtwo(f)$ is large. Hence, we ask if it is possible to improve~\cref{eq:deg:deg2} by showing $\deg(f) \le \alt(f)  \cdot \deg_{m'}(f) \cdot \degm(f) $ for any $m' \ne m$.

\end{description}

\subsection*{Acknowledgments} The authors would like to thank the anonymous reviewers for pointing out errors in the earlier version and for providing constructive comments which improved the presentation of the paper.

\bibliographystyle{plain}
\bibliography{references}
\ifthenelse{\equal{\movetoappendix}{1}}{
        \appendix
        \section{Appendix}
        \includecollection{appendix}
} { }

\end{document}